\documentclass[11pt]{article}
\usepackage[colorlinks,citecolor=blue,bookmarks=true]{hyperref}
\usepackage{amsmath} % AMS Math Package
\usepackage{amsthm} % Theorem Formatting
\usepackage{amssymb}	% Math symbols such as \mathbb
\usepackage{graphicx} % Allows for eps images
\usepackage{multirow}
\usepackage{color}
\usepackage[dvips,letterpaper,margin=1in]{geometry}
\usepackage[capitalize,noabbrev]{cleveref}
\usepackage{subcaption}

\usepackage[utf8]{inputenc}
\usepackage[english]{babel}

\usepackage{algorithm}
\usepackage{algpseudocode}
\makeatletter
\renewcommand{\ALG@name}{Algorithm}
\makeatother

\usepackage{bm}
\usepackage{bbm}
\usepackage{diagbox}
\usepackage{mathtools}
\usepackage{ytableau}

\newtheorem{theorem}{Theorem}[section]

\newtheorem{lemma}[theorem]{Lemma}
\newtheorem{corollary}[theorem]{Corollary}

\newtheorem{fact}[theorem]{Fact}

\newtheorem{definition}[theorem]{Definition}

\newtheorem{problem}[theorem]{Problem}

\crefname{question}{Question}{Questions}

\newtheorem{notation}[theorem]{Notation}

\DeclarePairedDelimiter\ket{\lvert}{\rangle}
\DeclarePairedDelimiter\bra{\langle}{\rvert}

\newcommand{\E}{\mathop{\bf E\/}}

\newcommand{\tr} {\operatorname{tr}}

\newcommand{\supp} {\operatorname{supp}}
\newcommand{\spanspace} {\operatorname{span}}

\newcommand{\ketbra}[2]{\ensuremath{\ket{#1}\!\bra{#2}}}

\newcommand{\kett}[1]{|#1\rangle\!\rangle}
\newcommand{\bbra}[1]{\langle\!\langle#1|}
\newcommand{\kettbbra}[2]{\ensuremath{\kett{#1}\!\bbra{#2}}}
\newcommand{\bbrakett}[2]{\ensuremath{\langle\!\langle{#1}\vert{#2}\rangle\!\rangle}}

\allowdisplaybreaks

\DeclarePairedDelimiter\parens{\lparen}{\rparen}

\newcommand{\calE}{\mathcal{E}}

\newcommand{\qchannel}{\textbf{\textup{QChan}}}
\newcommand{\isochannel}{\textbf{\textup{ISO}}}

\title{Optimal lower bound for quantum channel tomography\\
in away-from-boundary regime}
\author{
Kean Chen \thanks{University of Pennsylvania, Philadelphia, USA. Email: \texttt{keanchen.gan@gmail.com}}\and
Zhicheng Zhang \thanks{University of Technology Sydney, Sydney, Australia. Email: \texttt{iszczhang@gmail.com}}\and
Nengkun Yu \thanks{Stony Brook University, NY, USA. Email: \texttt{nengkunyu@gmail.com}}
}
\date{}

\begin{document}

\maketitle

\begin{abstract}
Consider quantum channels with input dimension $d_1$, output dimension $d_2$ and Kraus rank at most $r$.
Any such channel must satisfy the constraint $rd_2\geq d_1$, and the parameter regime $rd_2=d_1$ is called the boundary regime.
In this paper, we show an optimal query lower bound $\Omega(rd_1d_2/\varepsilon^2)$ for quantum channel tomography to within diamond norm error $\varepsilon$ in the away-from-boundary regime $rd_2\geq 2d_1$, matching the existing upper bound $O(rd_1d_2/\varepsilon^2)$.
In particular, this lower bound fully settles the query complexity 
for the commonly studied case of equal input and output dimensions $d_1=d_2=d$ with $r\geq 2$, in sharp contrast to the unitary case $r=1$ where Heisenberg scaling $\Theta(d^2/\varepsilon)$ is achievable.

%of quantum channel tomography when the input and output dimensions are equal (i.e., $d_1=d_2$) with non-unitary constraint (i.e., $r\geq 2$). This stands in sharp contrast to the unitary case (i.e., $r=1$), which exhibits Heisenberg scaling $\Theta(d_1d_2/\varepsilon)$.
\end{abstract}

\section{Introduction}

Estimating an unknown quantum physical process from experimental data is a foundational task in quantum computing and quantum information. A central question is to quantify the informational resources required for such estimation when the unknown process is given as a black-box quantum channel. In this paper, we study \textit{quantum channel tomography}: given query access to an unknown quantum channel $\calE$, %with input dimension $d_1$, output dimension $d_2$ and Kraus rank at most $r$, 
the goal is to learn a full classical description of $\calE$ using as few queries as possible (to a prescribed accuracy, e.g., in diamond norm).% to within diamond norm error $\varepsilon$.

Research on quantum channel tomography traces back to the more basic problem of quantum state tomography, which aims to learn a full classical description of an unknown quantum state from samples.
Quantum state tomography can be viewed as a special case of quantum channel tomography in which the input dimension is $1$.
The optimal tomography of pure states has been well understood since the seminal works~\cite{hayashi1998asymptotic,bruss1999optimal,keyl1999optimal}.
Optimal tomography of mixed states was developed later in~\cite{Haah_2017,10.1145/2897518.2897544} and subsequently refined, extended, and clarified in~\cite{o2017efficient,GKKT20,yuen2023improved,scharnhorst2025optimal,pelecanos2025debiased,pelecanos2025mixedstatetomographyreduces}. 
%an important special case when the channel $\calE$ has input dimension $1$.

Compared with quantum state tomography, general quantum channel tomography involves a richer set of considerations. One may design the input states arbitrarily (including entanglement with ancillas), apply the unknown channel sequentially and adaptively, and perform collective measurements across multiple uses, which leads to more subtle analyses.
Despite this difficulty, extensive work~\cite{chuang1997prescription,poyatos1997complete,leung2000towards,d2001quantum,mohseni2008quantum,kliesch2019guaranteed,bouchard2019quantum,surawy2022projected,Oufkir_2023,oufkir2023adaptivity,huang2023learning,pmlr-v195-fawzi23a,caro2024learning,rosenthal2024quantum,zhao2024learning,zambrano2025fast,yoshida2025quantum} has been devoted to quantum channel tomography over the last thirty years.
Notably, for tomography of unitary channels, Haah, Kothari, O'Donnell, and Tang~\cite{haah2023query} settled the optimal query complexity $\Theta(d^2/\varepsilon)$, where $d$ is the channel dimension and $\varepsilon$ is the target error in diamond norm.
For tomography of general channels using only non-adaptive incoherent measurements, Oufkir established a near-optimal query complexity  $\widetilde{\Theta}(d_1^3 d_2^3/\varepsilon^2)$~\cite{Oufkir_2023,oufkir2023adaptivity}, generalizing the algorithm in~\cite{surawy2022projected}, where $d_1$ and $d_2$ are the input and output dimensions and $\varepsilon$ is the diamond norm error. For isometry channel tomography, Yoshida, Miyazaki, and Murao~\cite{yoshida2025quantum} established a query lower bound of $\Omega((d_2-d_1)d_1 /(\varepsilon^{2}\log 1/\varepsilon))$.

In the most general setting, the unknown channel has input dimension $d_1$, output dimension $d_2$, and Kraus rank at most $r$. Recent work has substantially improved our understanding of the optimal scalings.
On the upper-bound side, Mele and Bittel~\cite{mele2025optimal} and Chen, Yu, and Zhang~\cite{chen2025quantum} showed that $O\parens*{rd_1d_2/\varepsilon^2}$ queries suffice for channel tomography with diamond norm error $\varepsilon$.
Moreover, in the boundary regime $rd_2=d_1$, \cite{chen2025quantum} showed that $O(rd_1d_2/\varepsilon)$ queries suffice for channel tomography with Choi-state trace norm error $\varepsilon$, achieving the Heisenberg scaling.

Girardi, Mele, Zhao, Fanizza, and
Lami~\cite{GMZFL25} and Yoshida, Niwa, and Murao~\cite{yoshida2025random} then algorithmically strengthened the local test technique in~\cite{chen2025quantum} by explicitly constructing random Stinespring dilation superchannels.
Conceptually, local test and random dilation for channels can be viewed as dual techniques in the Heisenberg and Schr{\"o}dinger pictures,  respectively. The underlying ideas trace back to local test and random purification for quantum states~\cite{tang2025conjugate,chen2024local,soleimanifar2022testing}.
More developments can be found in~\cite{pelecanos2025mixedstatetomographyreduces,girardi2025random,mele2025random,walter2025random}.

On the lower-bound side, it was shown in \cite{GMZFL25} that $\Omega(rd_1d_2)$ queries are required for channel tomography at constant error, improving the prior lower bound $\Omega(d_1^2 d_2^2/\log(d_1d_2))$ for full Kraus-rank (i.e., $r=d_1d_2$) tomography due to Rosenthal, Aaronson, Subramanian, Datta, and Gur~\cite{rosenthal2024quantum}. More recently, Oufkir and 
Girardi~\cite{oufkir2026improved} incorporated the $\varepsilon$-dependence and proved a lower bound of $\Omega\parens*{rd_1d_2/(\varepsilon^2\log (d_2r/\varepsilon))}$\footnote{
    During the preparation of this manuscript, we became aware of a very recent update (arXiv v3) of \cite{oufkir2026improved}, in which their lower bound in the away-from-boundary regime $rd_2\geq 2d_1$ is improved to $\Omega\parens*{rd_1d_2/(\varepsilon^2\log (d_2r/\varepsilon))}$.
    See \Cref{sub:related-work} for further discussion.
} 
in the away-from-boundary regime $ rd_2\geq 2d_1$, matching the upper bound in \cite{mele2025optimal,chen2025quantum} up to a logarithmic factor. They also showed a lower bound $\Omega\parens*{rd_1d_2/(\varepsilon\log (d_2r/\varepsilon))}$ in the boundary regime $rd_2=d_1$, where $\varepsilon$ can be either Choi-state trace norm or diamond norm error, matching the upper bound in \cite{chen2025quantum} up to a logarithmic factor.

A remaining open question is whether a matching query lower bound $\Omega(rd_1d_2/\varepsilon^2)$ for general quantum channel tomography can be proved without logarithmic factors. In this paper, we resolve this question by establishing such a lower bound in the away-from-boundary regime $ rd_2 \geq 2d_1$. In particular, this settles the most commonly studied case of equal input and output dimensions $d_1=d_2$ with $r\geq 2$.

\subsection{Main results}

%\begin{table}[h]
%\centering
%\begin{tabular}{|c|c|c|c|}
%\hline
%             & \begin{tabular}{@{}c@{}}Boundary \\ $rd_2=d_1$\end{tabular} & \begin{tabular}{@{}c@{}} Near-boundary \\$d_1< rd_2\leq 2d_1$\end{tabular} &  \begin{tabular}{@{}c@{}} Away-from-boundary\\ $2d_1\leq rd_2$\end{tabular} \\ \hline
%Upper bounds &    $O\!\left(\dfrac{rd_1d_2}{\varepsilon}\right)$~{\small\cite{chen2025quantum}}  &    $O\parens*{\dfrac{rd_1d_2}{\varepsilon^2}}$~{\small \cite{mele2025optimal,chen2025quantum}}                                       &    $O\parens*{\dfrac{rd_1d_2}{\varepsilon^2}}$~{\small\cite{mele2025optimal,chen2025quantum}}                                       \\ \hline
%Lower bounds &     $\Omega\parens*{\dfrac{rd_1d_2}{\varepsilon\log (d_2r/\varepsilon)}}$~{\small\cite{oufkir2026improved}}                       &      $\Omega\parens*{rd_1d_2}$~{\small\cite{GMZFL25}}                                     &                                   $\Omega\parens*{\dfrac{rd_1d_2}{\varepsilon^2}}$ \textbf{This work}                   \\ \hline
%\end{tabular}
%\end{table}

Our main result is as follows.
\begin{theorem}[Optimal lower bound in away-from-boundary regime, \cref{thm-1110326} restated]
\label{thm:main}
Let $d_1,d_2,r$ be positive integers such that $ rd_2\geq 2d_1$. Tomography of quantum channels with input dimension $d_1$, output dimension $d_2$ and Kraus rank at most $r$, and to within diamond norm error $\varepsilon$, requires 
$\Omega(rd_1d_2/\varepsilon^2)$
queries.
\end{theorem}
\Cref{thm:main} is based on a tight analysis on sets of isometries with a specific structure (which we call the ``hard'' isometry set), using the formalism of quantum combs and testers.
Recent work~\cite{oufkir2026improved} provides an instantiation of such a ``hard'' isometry set with sufficiently large cardinality and the desired separation properties, allowing our analysis to apply to their construction and yield the optimal lower bound.

As a special case of our main result, we consider quantum channels with equal input and output dimensions, i.e., $d_1 = d_2 = d$, which are simply called $d$-dimensional quantum channels.
Combined with known results on unitary tomography~\cite{haah2023query} and upper bound on quantum channel tomography~\cite{mele2025optimal,chen2025quantum}, we can fully settle the query complexity for the tomography task of $d$-dimensional quantum channels.

\begin{corollary}[Tomography of $d$-dimensional quantum channels]\label{equaldimension}
The query complexity for tomography of $d$-dimensional quantum channels
$\mathcal{E}$ with Kraus rank at most $r$, and to within diamond-norm error $\varepsilon$, is
\[
\Theta\!\left(\frac{rd^2}{\varepsilon^{\min\{r,2\}}}\right).
\]
\end{corollary}
Note that this reveals a sharp phase transition in the dependence on $\varepsilon$: it exhibits Heisenberg scaling $1/\varepsilon$ when $r=1$ and classical scaling $1/\varepsilon^2$ when $r\geq 2$.

As another special case, we consider tomography of quantum channels with input dimension $1$, which reduce to quantum state tomography.
Then, we can reproduce the recent development of the optimal sample lower bound for quantum state tomography~\cite{scharnhorst2025optimal}, which matches the known upper bound~\cite{10.1145/2897518.2897544}.
\begin{corollary}[State tomography]
Tomography of a $d$-dimensional mixed state with rank at most $r$, to within trace norm error $\varepsilon$, requires $\Omega(dr/\varepsilon^2)$ samples.
\end{corollary}
Our method for this lower bound is very different from that in \cite{scharnhorst2025optimal}, which may be of independent interest.

Then, we summarize the current best upper and lower bounds for quantum channel tomography in different parameter regimes in~\Cref{tab:comparison}.

\begin{table}[h]
\centering
% Symmetric outer-row padding (height = depth = 1.6ex). Adjust to taste.
\newcommand{\outerstrut}{\rule[-2.8ex]{0pt}{7ex}}

\begin{tabular}{|c|c|c|c|}
\hline
\outerstrut
             & \begin{tabular}{@{}c@{}}Boundary$^{*}$ \\ $rd_2=d_1$\end{tabular}
             & \begin{tabular}{@{}c@{}} Near-boundary \\$d_1< rd_2< 2d_1$\end{tabular}
             & \begin{tabular}{@{}c@{}} Away-from-boundary\\ $ rd_2 \geq 2d_1$\end{tabular} \\ \hline
Upper bounds\outerstrut
            & $O\!\left(\dfrac{rd_1d_2}{\varepsilon}\right)$~\cite{chen2025quantum}
            & \multicolumn{2}{c|}{$O\parens*{\dfrac{rd_1d_2}{\varepsilon^2}}$~\cite{mele2025optimal,chen2025quantum}} \\ \hline
Lower bounds\outerstrut
            & $\Omega\parens*{\dfrac{rd_1d_2}{\varepsilon\log (d_2r/\varepsilon)}}$~\cite{oufkir2026improved}
            & $\Omega\parens*{rd_1d_2}$~\cite{GMZFL25}
            & $\Omega\parens*{\dfrac{rd_1d_2}{\varepsilon^2}}$ \textbf{This work} \\ \hline
\end{tabular}
\caption{Upper and lower bounds for quantum channel tomography in different parameter regimes. Note that $rd_2\geq d_1$ holds for any quantum channels. There is a phase transition from the boundary regime to away-from-boundary regime: the Heisenberg scaling $1/\varepsilon$ becomes the classical scaling $1/\varepsilon^2$. 
\\
$*$: {\footnotesize In the boundary regime of this table, the upper bound holds for Choi-state trace norm error and the lower bound hold for both Choi-state trace norm and diamond norm errors. All other bounds hold for diamond-norm error.} }
\label{tab:comparison}
\end{table}

\subsection{Related work}
\label{sub:related-work}

During the preparation of this manuscript, we became aware of a very recent update (arXiv v$3$) of \cite{oufkir2026improved}, in which their lower bound in the away-from-boundary regime $ rd_2\geq 2d_1$ is improved from $\Omega\parens*{rd_1d_2/(\varepsilon\log (d_2r/\varepsilon))}$ to $\Omega\parens*{rd_1d_2/(\varepsilon^2\log (d_2r/\varepsilon))}$, which also achieves classical scaling $1/\varepsilon^2$ and matches the upper bound up to a logarithmic factor $\log (d_2r/\varepsilon)$. Their proof relies on an information-theoretic approach, which is different from our approach for proving the lower bound.
Specifically, in our proof of \Cref{thm:main}, we provide a tight analysis for the hardness of discriminating a specific isometry family (see \cref{thm-1110122}), then we combine the isometry net instantiation provided in~\cite[arXiv v$1$]{oufkir2026improved} with our hardness result (i.e., \cref{thm-1110122}), to obtain the optimal lower bound without logarithmic factors. 
In contrast to approaches using information-theoretic tools, our analysis is based on the formalism of quantum combs and testers.

\subsection{Discussion}

This work establishes a matching lower bound of $\Omega(rd_1d_2/\varepsilon^2)$ on the number of queries needed for quantum channel tomography in the away-from-boundary regime $ rd_2\geq 2d_1$.
Combined with the prior upper bound $O\parens*{rd_1d_2/\varepsilon^2}$~\cite{mele2025optimal,chen2025quantum}, our new query lower bound fully settles the most commonly studied case of equal input and output dimensions $d_1=d_2=d$ with $r\geq 2$.
This optimal scaling is in sharp contrast to the Heisenberg scaling $\Theta(d^2/\varepsilon)$ in unitary channel tomography.

An important open question is whether one can further settle the query complexity for quantum channel tomography beyond the away-from-boundary regime $ rd_2\geq 2d_1$.

\section{Preliminaries}
\subsection{Notation}
We use $\mathcal{L}(\mathcal{H})$ to denote the set of linear operators on the Hilbert space $\mathcal{H}$. Given two orthonormal bases for $\mathcal{H}_0$ and $\mathcal{H}_1$ respectively, we can represent each linear operator from $\mathcal{H}_0$ to $\mathcal{H}_1$ by a $\dim(\mathcal{H}_1)\times \dim(\mathcal{H}_0)$ matrix and for such a matrix \(X\), we use \(\kett{X}\in \mathcal{H}_1\otimes\mathcal{H}_0\) to denote the vector obtained by flattening the matrix $X$. It is easy to see the following facts:
\[\kett{\ketbra{\psi}{\phi}}=\ket{\psi}\ket{\phi^*}, \quad\quad\quad \kett{XYZ}=X\otimes Z^\textup{T} \kett{Y},\]
where \(\ket{\phi^*}\) is the entry-wise complex conjugate of \(\ket{\phi}\) w.r.t. to a given orthonormal basis, and $Z^\textup{T}$ is the transpose of the matrix $Z$. 
The inner product can be denoted by \(\bbrakett{X}{Y}=\tr(X^\dag Y)\).
For two linear operators $X,Y$, we use $X\sqsubseteq Y$ to denote that $Y-X$ is positive semidefinite.

Let $n, m, d$ be positive integers such that $n\geq m$. Let $\mathcal{H}_1\cong\cdots\cong \mathcal{H}_n\cong\mathbb{C}^{d}$ be $n$ copies of the $d$-dimensional Hilbert space. 
Let $S\subseteq [n]=\{1,2,\ldots,n\}$ be a set of integers and $\ket{\psi}\in\mathbb{C}^d$ be a state. We use the following notation 
\[\ket{\psi}^{\otimes S}\]
to denote the state $\ket{\psi}^{\otimes |S|}$ on $\bigotimes_{i\in S} \mathcal{H}_i$.
Therefore, if $\ket{\varphi}\in \mathbb{C}^d$ is another state, then
\[\ket{\psi}^{\otimes S} \otimes\ket{\varphi}^{\otimes [n]\setminus S}\]
denotes the state $\bigotimes_{i=1}^n \ket{x_i}$ on $\bigotimes_{i=1}^n \mathcal{H}_i$ where $\ket{x_i}=\ket{\psi}$ for $i\in S$, and $\ket{x_i}=\ket{\varphi}$ otherwise.

\subsection{Quantum channels}

A quantum channel with input dimension $d_1$ and output dimension $d_2$ is described by a linear map $\mathcal{E}:\mathcal{L}(\mathbb{C}^{d_1})\rightarrow\mathcal{L}(\mathbb{C}^{d_2})$ such that $\mathcal{E}$ is completely positive and trace-preserving (see, e.g., \cite{NC10,watrous2018theory,hayashi2017quantum}). 
%A quantum channel has various types of representations, such as the Kraus representation and Choi-Jamio{\l}kowski representation (see, e.g., \cite{watrous2018theory,hayashi2017quantum}).

In the Kraus representation~\cite{kraus1983states}, a quantum channel $\mathcal{E}$ is written as
\[\mathcal{E}(\rho)=\sum_{i=1}^r E_i \rho E_i^\dag,\]
where $E_i: \mathbb{C}^{d_1}\rightarrow\mathbb{C}^{d_2}$ are non-zero linear operators that satisfy $\sum_{i=1}^r E_i^\dag E_i=I$, which are called Kraus operators. We can always find a set of $E_i$ such that $\tr(E_i^\dag E_j)=0$ for $i\neq j$, then those $E_i$ are called orthogonal Kraus operators and $r$ is called the \textit{Kraus rank}. 
Note that $r$ must satisfy $d_1/d_2\leq r\leq d_1d_2$.
A quantum channel that has Kraus rank $r=1$ is an isometry channels $\mathcal{V}=V(\cdot) V^\dag$, where $V:\mathbb{C}^{d_1}\rightarrow\mathbb{C}^{d_2}$ is an isometry operator, i.e., $V^\dag V=I_{d_1}$, and it must hold that $d_2\geq d_1$.

\begin{notation}
We use $\qchannel_{d_1,d_2}^r$ to denote the set of all quantum channels $\mathcal{E}:\mathcal{L}(\mathbb{C}^{d_1})\rightarrow \mathcal{L}(\mathbb{C}^{d_2})$ that have Kraus rank at most $r$.
In particular, we use $\isochannel_{d_1,d_2}$ to denote the set of isometry channels with input dimension $d_1$ and output dimension $d_2$, which is equivalent to $\qchannel_{d_1,d_2}^1$. %\footnote{The isometry exists only if $d_1\leq d_2$.}
\end{notation}

In the Choi-Jamio{\l}kowski representation~\cite{choi1975completely,jamiolkowski1972linear,jamiolkowski1972linear}, $\mathcal{E}$ is represented by the Choi-Jamio{\l}kowski operator 
\[C_\mathcal{E}=(\mathcal{E}\otimes \mathcal{I})(\kettbbra{I}{I})\in\mathcal{L}(\mathbb{C}^{d_2}\otimes\mathbb{C}^{d_1}),\]
where $\kett{I}=\sum\limits_{i}\ket{i}\ket{i}\in \mathbb{C}^{d_1}\otimes \mathbb{C}^{d_1}$ is an unnormalized maximally entangled state. We may simply call it the Choi operator. 
Note that we can write $C_\mathcal{E}=\sum_{i=1}^r\kettbbra{E_i}{E_i}$, where $E_i$ are orthogonal Kraus operators and thus $\kett{E_i}$ are pairwise orthogonal vectors. Therefore, the Kraus rank equals the rank of the Choi operator.

\paragraph{Stinespring dilation.}
Using the Stinespring dilation~\cite{stinespring1955positive}, we can also write a quantum channel $\mathcal{E}$ with Kraus operators $\{E_i\}_{i=1}^r$ as
\begin{equation}\label{eq-1230155}
\mathcal{E}(\cdot)=\tr_{\mathcal{H}_\mathrm{anc}}(V(\cdot) V^\dag),
\end{equation}
where $\mathcal{H}_\mathrm{anc}\cong \mathbb{C}^{r}$ and $V=\sum_{i=1}^r\ket{i}_\mathrm{anc}\otimes E_i$ is an isometry operator. 
By this, one can notice that $rd_2\geq d_1$ must hold.
An isometry channel $\mathcal{V}=V(\cdot)V^\dag$ that satisfies \cref{eq-1230155} is called a dilation of $\mathcal{E}$. 
Suppose $\mathcal{V}_1$ is a dilation of $\mathcal{E}$, then $\mathcal{V}_2$ is a dilation of $\mathcal{E}$ if and only if they differ by a unitary on $\mathcal{H}_\mathrm{anc}$, i.e., $V_2=(U\otimes I_{d_2}) V_1$ for $U:\mathcal{H}_\mathrm{anc}\rightarrow\mathcal{H}_\mathrm{anc}$ a unitary.
Conversely, given an isometry $V\in \isochannel_{d_1,rd_2}$, the channel $\mathcal{E}(\cdot)=\tr_{r}(V(\cdot) V^\dag)$ obtained from $V$ by tracing out an $r$-dimensional subsystem has Kraus rank at most $r$.

\subsection{Quantum combs and testers}\label{sec-6290132}
The quantum comb~\cite{chiribella2008quantum,chiribella2009theoretical} is a powerful tool to describe (higher) transformations of quantum processes. Specifically, the Choi-Jamio{\l}kowski representation of quantum channels (i.e., transformations of quantum states) can be generalized to a higher-level concept (i.e., transformations of quantum processes), which is called \textit{quantum comb}.
\begin{definition}[Quantum comb~\cite{chiribella2009theoretical}]\label{def-681501}
For an integer $n\geq 1$, a quantum $n$-comb defined on a sequence of $2n$ Hilbert spaces $(\mathcal{H}_0,\mathcal{H}_1,\ldots,\mathcal{H}_{2n-1})$ is a positive semidefinite operator $X$ on $\bigotimes_{j=0}^{2n-1} \mathcal{H}_{j}$ such that there exists a sequence of operators $X^{(n)}, X^{(n-1)},\ldots, X^{(1)}, X^{(0)}$ such that
\begin{equation}\label{eq-681546}
\begin{split}
\tr_{\mathcal{H}_{2j-1}}\!\left(X^{(j)}\right)&=I_{\mathcal{H}_{2j-2}}\otimes X^{(j-1)},\quad 1\leq j \leq n,  
\end{split}
\end{equation}
where $X^{(n)}=X$ and $X^{(0)}=1$.
\end{definition}

We can easily see the following facts:
A quantum $1$-comb is simply the Choi-Jamio{\l}kowski operator of a quantum channel.
Any convex combination of quantum $n$-combs is also a quantum $n$-comb.

Then, we introduce the link product ``$\star$''.
\begin{definition}[Link product ``$\star$''~\cite{chiribella2008quantum,chiribella2009theoretical}]
\label{def-720255}
Suppose $X$ is a linear operator on $\mathcal{H}_{\bm{i}}=\mathcal{H}_{i_1}\otimes\mathcal{H}_{i_2}\otimes\cdots\otimes\mathcal{H}_{i_{n}}$ and $Y$ is a linear operator on $\mathcal{H}_{\bm{j}}=\mathcal{H}_{j_1}\otimes\mathcal{H}_{j_2}\otimes\cdots\otimes\mathcal{H}_{j_{m}}$,
where $\bm{i}=(i_1,\ldots,i_n)$ is a sequence of pairwise distinct indices, and likewise for $\bm{j}=(j_1,\ldots,j_m)$.
Let $\bm{a}=\bm{i}\cap\bm{j}$ be the set of indices in both $\bm{i}$ and $\bm{j}$ and $\bm{b}=\bm{i}\cup\bm{j}$ be the set of indices in either $\bm{i}$ or $\bm{j}$.
Then, the combination of $X$ and $Y$ is defined by
\[X\star Y= \tr_{\mathcal{H}_{\bm{a}}}\!\left(X^{\textup{T}_{\mathcal{H}_{\bm{a}}}} \cdot Y\right)=\tr_{\mathcal{H}_{\bm{a}}}\!\left(X\cdot Y^{\textup{T}_{\mathcal{H}_{{\bm{a}}}}}\right),\]
where $\mathcal{H}_{\bm{a}}$ means the tensor product of subsystems labeled by the indices in $\bm{a}$, $\textup{T}_{\mathcal{H}_{\bm{a}}}$ means the partial transpose on $\mathcal{H}_{\bm{a}}$, both $X$ and $Y$ are treated as linear operators on $\mathcal{H}_{\bm{b}}$, extended by tensoring with the identity operator as needed.
\end{definition}

The link product describes the combination of quantum combs. For example, suppose $X$ is an $n$-comb on $(\mathcal{H}_0,\mathcal{H}_1,\ldots,\mathcal{H}_{2n-1})$ and $Y$ is an $(n-1)$-comb on $(\mathcal{H}_1,\mathcal{H}_2\,\ldots,\mathcal{H}_{2n-2})$, then
\begin{align}
X\star Y= \tr_{\mathcal{H}_{1:2n-2}}\!\left(X^{\textup{T}_{\mathcal{H}_{1:2n-2}}}\cdot (I_{\mathcal{H}_{2n-1}}\otimes Y \otimes I_{\mathcal{H}_0})\right)=\tr_{\mathcal{H}_{1:2n-2}}\!\left(X\cdot (I_{\mathcal{H}_{2n-1}}\otimes Y^{\textup{T}} \otimes I_{\mathcal{H}_0})\right)\nonumber
\end{align}
turns out to be a $1$-comb on $(\mathcal{H}_0,\mathcal{H}_{2n-1})$.
The link product also has many good properties.
It preserves the L\"owner order: if $X,Y\sqsupseteq 0$ then $X\star Y\sqsupseteq 0$~\cite[Theorem 2]{chiribella2009theoretical}. 
It is commutative $X\star Y=Y \star X$, and associative $(X\star Y)\star Z=X\star (Y\star Z)$ whenever $X,Y,Z$ do not share a common subsystem (i.e., there is no subsystem that is a subsystem of all three).

%Moreover, it characterizes the channel concatenation under the Choi representation: given two quantum channels $\mathcal{E}_1:\mathcal{L}(\mathcal{H}_1)\rightarrow\mathcal{L}(\mathcal{H}_2)$ and $\mathcal{E}_2:\mathcal{L}(\mathcal{H}_2)\rightarrow\mathcal{L}(\mathcal{H}_3)$, we have 
%$C_{\mathcal{E}_2\circ \mathcal{E}_1}=C_{\mathcal{E}_2}\star C_{\mathcal{E}_1}$, where $C_\mathcal{E}$ denotes the Choi operator of $\mathcal{E}$.

\subsubsection{Quantum channel testers} 
A \textit{quantum channel tester} means a quantum algorithm that can make multiple queries to an unknown quantum channel and then produces a classical output.
We adopt the quantum tester formalism based on Choi-Jamio{\l}kowski representation (see, e.g., \cite{chiribella2009theoretical,bavaresco2021strict,bavaresco2022unitary}), which provides a practical framework for studying various classes of quantum testers, such as parallel and sequential ones.

Suppose a quantum channel tester uses $n$ queries to an unknown quantum channel $\mathcal{E}$. 
We label the input and output systems of the $i$-th query to $\mathcal{E}$ as $\mathcal{H}_{\mathrm{A},i}$ and $\mathcal{H}_{\mathrm{B},i}$, i.e., the $i$-th copy of the unknown channel is a linear map from $\mathcal{L}(\mathcal{H}_{\mathrm{A},i})$ to $\mathcal{L}(\mathcal{H}_{\mathrm{B},i})$.

In a sequential tester, one sends a quantum system through the first use of the channel $\mathcal{E}$ and then feeds the resulting output into subsequent uses, potentially along with ancillary systems, while allowing arbitrary CPTP maps to act between uses of $\mathcal{E}$. After all $n$ uses of the channel $\mathcal{E}$, a POVM is performed on the final output state. In other words, sequential testers can represent coherent and adaptive query-access algorithms.

\begin{definition}[Sequential tester]
A sequential tester that uses $n$ queries to an unknown channel is a set of linear operators $\{T_i\}_i$ for $T_i\in\mathcal{L}(\bigotimes_{j=1}^n \mathcal{H}_{\mathrm{A},j}\otimes\mathcal{H}_{\mathrm{B},j})$ such that $T_i\sqsupseteq 0$ and $\sum_i T_i$ is a quantum $(n+1)$-comb on $(\mathcal{H}_0,\mathcal{H}_{\mathrm{A},1},\mathcal{H}_{\mathrm{B},1},\ldots,\mathcal{H}_{\mathrm{A},n},\mathcal{H}_{\mathrm{B},n},\mathcal{H}_{n+1})$, where $\mathcal{H}_0\cong\mathcal{H}_{n+1}\cong\mathbb{C}$ are one-dimensional.
\end{definition}
It is known that any sequential tester can be realized by a sequential query-access algorithm and any sequential query-access algorithm can be described by a sequential tester~\cite{chiribella2009theoretical,bavaresco2022unitary}.
When we apply a sequential tester $\{T_i\}_i$ to $n$ queries to a quantum channel $\mathcal{E}$, we get the classical outcome $i$ with probability 
\begin{equation*}
p_i=T_i\star C_\mathcal{E}^{\otimes n}= \tr(T_i (C_{\mathcal{E}}^{\otimes n})^\mathrm{T})=\tr(T_i^\mathrm{T} C_{\mathcal{E}}^{\otimes n}),
\end{equation*}
where $C_\mathcal{E}^{\otimes n}$
is the Choi operator of all $n$ queries to the channel $\mathcal{E}$ and $(\cdot)^{\mathrm{T}}$ denotes matrix transposition.

\subsubsection{Discrimination of quantum channels}
Suppose $\mathcal{N}$ is a finite set of quantum channels. Then, the discrimination problem for channels in $\mathcal{N}$ is defined as follows.
\begin{problem}
Suppose $\mathcal{E}$ is uniformly randomly chosen from the set $\mathcal{N}$. The algorithm (or tester) can make $n$ queries to the channel $\mathcal{E}$ and the goal is to identify $\mathcal{E}$.
\end{problem}
Suppose $\{T_\mathcal{E}\}_{\mathcal{E}\in\mathcal{N}}$ is a sequential tester for this discrimination task where $T_\mathcal{E}$ corresponds to outputting the label $\mathcal{E}$. Then, the success probability can be expressed as
\[\Pr[\textup{success}]=\frac{1}{|\mathcal{N}|} \sum_{\mathcal{E}\in\mathcal{N}} T_\mathcal{E}\star C_{\mathcal{E}}^{\otimes n},\]
where $C_{\mathcal{E}}$ denotes the (unnormalized) Choi state of $\mathcal{E}$.
We say an algorithm solves the discrimination problem if the success probability is higher than $2/3$.

\section{Hardness of discriminating isometries}\label{sec-1110117}
\subsection{Hard instance}\label{sec-1110250}
Suppose $d_1,d_2$ are positive integers such that $d_2\geq 2d_1$ and $\varepsilon\in (0,1)$.
Define the Hilbert space $\mathcal{H}_{\mathrm{A}}\cong \mathbb{C}^{d_1}$ with an orthonormal basis $\{\ket{1}_\mathrm{A},\ldots,\ket{d_1}_\mathrm{A}\}$ and $\mathcal{H}_{\mathrm{B}}\cong \mathbb{C}^{d_2}$ with an orthonormal basis $\{\ket{1}_\mathrm{B},\ldots,\ket{d_2}_\mathrm{B}\}$.
Define the isometries $V_0,\Delta:\mathcal{H}_{\mathrm{A}}\rightarrow \mathcal{H}_{\mathrm{B}}$ as follows
\[V_0\coloneqq \sum_{i=1}^{d_1}\ket{i}_\mathrm{B}\bra{i}_\mathrm{A},\quad\textup{and}\quad \Delta\coloneqq \sum_{i=1}^{d_1}\ket{d_1+i}_\mathrm{B}\bra{i}_\mathrm{A}.\]
Then, for any $U\in\mathbb{U}_{d_2-d_1}$, we define the isometry $V_{\varepsilon,U}:\mathcal{H}_{\mathrm{A}}\rightarrow\mathcal{H}_{\mathrm{B}}$ as
\begin{equation}\label{eq-140115}
\begin{split}
V_{\varepsilon,U}& \coloneqq (I_{d_1}\oplus U)\left(\sqrt{1-\varepsilon^2}V_0+\varepsilon \Delta\right) \\
&= \sqrt{1-\varepsilon^2}V_0 + \varepsilon U\Delta,
\end{split}
\end{equation}
where $I_{d_1}=\sum_{i=1}^{d_1}\ket{i}_{\mathrm{B}}\bra{i}_{\mathrm{B}}$, and $U\in\mathbb{U}_{d_2-d_1}$ acts on the subspace spanned by $\{\ket{d_1+1}_\mathrm{B},\ldots,\ket{d_2}_\mathrm{B}\}$.
Then, any subset of 
\[\{V_{\varepsilon,U} \,|\, U\in\mathbb{U}_{d_2-d_1}\}\]
is called a \textit{``hard''} isometry set.

Note that in the above construction, the orthonormal bases of $\mathcal{H}_\mathrm{A}$ and $\mathcal{H}_\mathrm{B}$ are chosen arbitrarily. Therefore, this construction can also be described in an abstract way.
\begin{definition}\label{def-1130036}
Let $\mathcal{H}_{\mathrm{A}}\cong \mathbb{C}^{d_1}$, $\mathcal{H}_{\mathrm{B}}\cong \mathbb{C}^{d_2}$.
Let $V_0:\mathcal{H}_\mathrm{A}\rightarrow\mathcal{H}_\mathrm{B}$ be an arbitrary but fixed isometry and let $\mathcal{H}_0$ be the image of $V_0$. Then, any subset of 
\[\left\{\sqrt{1-\varepsilon^2} V_0+\varepsilon \Delta\,\,\Big|\,\, \Delta:\mathcal{H}_\mathrm{A}\rightarrow\mathcal{H}_{0}^\perp \textup{ is an isometry}\right\}\]
is called a ``hard'' isometry set. 
\end{definition}

\subsection{Hardness of the discrimination problem}
Then, we have the following result.
\begin{theorem}\label{thm-1110122}
Suppose $\mathcal{N}$ is a finite ``hard'' isometry set (see \cref{sec-1110250}) with cardinality $|\mathcal{N}|\geq \exp(Cd_1d_2)$ for a universal constant $C$.
Then, any algorithm that solves the discrimination problem for the isometries in $\mathcal{N}$ requires at least $n\geq \Omega(d_1d_2/\varepsilon^2)$ queries.
\end{theorem}
\begin{proof}
Without loss of generality, we can assume $\mathcal{N}$ is a finite subset of $\{V_{\varepsilon,U}\,|\, U\in\mathbb{U}_{d_2-d_1}\}$ for $V_{\varepsilon,U}$ defined in \cref{eq-140115}.
Let $B= 2e^4$ be a constant. Suppose there is an algorithm that solves the discrimination problem using $n$ queries.
If $n > \frac{1}{B}d_1d_2/\varepsilon^2$, there is nothing to prove.
Otherwise we assume $n\leq \frac{1}{B}d_1d_2/\varepsilon^2$.

Note that each element $V$ in $\mathcal{N}$ is of the form
\begin{align}
V&= \sqrt{1-\varepsilon^2} V_0 + \varepsilon U\Delta,\nonumber 
\end{align}
where $U\in \mathbb{U}_{d_2-d_1}$.

Suppose the algorithm for distinguishing the net is described by a tester $\{T_V\}_{V\in\mathcal{N}}$.
Then, the success probability is
\begin{align}
\Pr[\textup{success}]&=\frac{1}{|\mathcal{N}|}\cdot \sum_{V\in\mathcal{N}} T_V \star \kettbbra{V}{V}^{\otimes n} \nonumber \\
&\leq \exp(-Cd_1d_2)\cdot \sum_{V\in\mathcal{N}} T_V \star \kettbbra{V}{V}^{\otimes n}. \nonumber
\end{align}
Here, $\kettbbra{V}{V}^{\otimes n}$ is an $n$-comb on $(\mathcal{H}_{\mathrm{A},1},\mathcal{H}_{\mathrm{B},1},\ldots,\mathcal{H}_{\mathrm{A},n},\mathcal{H}_{\mathrm{B},n})$, where $\mathcal{H}_{\mathrm{A},j}$ and $\mathcal{H}_{\mathrm{B},j}$ denote the input and output spaces of the $j$-th query to $V$, respectively.
Note that for $V\in\mathcal{N}$,
\[\kett{V}^{\otimes n}=\sum_{i=0}^n \left(\sqrt{1-\varepsilon^2}\right)^{n-i}\varepsilon^i  \sum_{\substack{S\subseteq[n]\\ |S|=i}} \kett{V_0}^{\otimes [n]\setminus S}\otimes \kett{U\Delta}^{\otimes S}.\]
For $i\in \{0,1,\ldots,n\}$, we define the state
\begin{equation}\label{eq-12270131}
\ket{\gamma_i}\coloneqq \frac{1}{\sqrt{\binom{n}{i}}} \sum_{\substack{S\subseteq [n]\\ |S|=i}}\kett{V_0}^{\otimes [n]\setminus S}\otimes \kett{\Delta}^{\otimes S}.
\end{equation}
Note that $\ket{\gamma_i}$ are pairwise orthogonal, and for any $V\in\mathcal{N}$, there exists a $U\in\mathbb{U}_{d_2-d_1}$ such that
\begin{equation}\label{eq-12271523}
\kett{V}^{\otimes n}=\sum_{i=0}^n\left(\sqrt{1-\varepsilon^2}\right)^{n-i}\varepsilon^i \sqrt{\binom{n}{i}} U^{\otimes n}\ket{\gamma_i},
\end{equation}
where $U^{\otimes n}$ acts as $(I_{d_1}\oplus U)^{\otimes n}$ on $\bigotimes_{i=1}^n \mathcal{H}_{\mathrm{B},i}$.
Next, we define the operator $\Gamma_i$ as
\begin{equation}\label{eq-12272104}
\Gamma_i\coloneqq \E_{U\sim\mathbb{U}_{d_2-d_1}}\left[U^{\otimes n}\ketbra{\gamma_i}{\gamma_i} U^{\dag\otimes n}\right],
\end{equation}
where $U^{\otimes n}$ acts as $(I_{d_1}\oplus U)^{\otimes n}$ on $\bigotimes_{i=1}^n \mathcal{H}_{\mathrm{B},i}$.
Note that $\supp(\Gamma_i)$ are also pairwise orthogonal, which can be easily seen from the fact that $\ket{\gamma_i}$ contains different number of $\kett{V_0}$ for different $i$.

If we can find some positive numbers $\lambda_0,\ldots,\lambda_n$ such that, for any $V\in\mathcal{N}$,
\begin{equation}\label{eq-12271213}
\kettbbra{V}{V}^{\otimes n}\sqsubseteq \sum_{i=0}^n \lambda_i \Gamma_i,
\end{equation}
then the success probability can be upper bounded as
\begin{align}
\Pr[\textup{success}]&\leq \exp(-C d_1d_2)\cdot \sum_{V\in \mathcal{N}} T_V\star \kettbbra{V}{V}^{\otimes n} \nonumber \\
&\leq \exp(-C d_1d_2)\cdot \sum_{V\in \mathcal{N}} T_V\star \sum_{j=0}^n \lambda_j\Gamma_j \nonumber \\
&= \exp(-C d_1d_2)\cdot \sum_{k=0}^n \lambda_k \cdot \sum_{V\in \mathcal{N}} T_V\star \sum_{j=0}^n \frac{\lambda_j}{\sum_{k=0}^n\lambda_k}\Gamma_j \nonumber \\
&= \exp(-Cd_1d_2)\cdot \sum_{i=0}^n \lambda_i, \label{eq-12272057}
\end{align}
where \cref{eq-12272057} is because 
\begin{itemize}
\item $\sum_{V\in\mathcal{N}} T_V$ is an (n+1)-comb with input and output dimensions $1$, and
\item $\sum_{j=0}^n \frac{\lambda_j}{\sum_{k=0}^n\lambda_k}\Gamma_j$ is an $n$-comb since $\Gamma_i$ is an $n$-comb (due to \cref{lemma-12270148}) and convex combination of $n$-combs is also an $n$-comb, 
\end{itemize}
so that their contraction evaluates to $1$.

By \cref{lemma-12281107}, there are $\{\lambda_i\}_{i=0}^n$ with $\sum_{i=0}^n\lambda_i\leq 3d_1^2d_2^2\exp(\sqrt{8n\varepsilon^2d_1d_2})$. Therefore, we know that the success probability can be bounded by 
\[\Pr[\textup{success}]\leq \exp(-Cd_1d_2)\cdot 3d_1^2d_2^2\exp\!\left(\sqrt{8n\varepsilon^2d_1d_2}\right).\]
If we want the success probability being at least $2/3$, we have
\[\sqrt{8n\varepsilon^2d_1d_2}\geq Cd_1d_2-2\ln(d_1d_2)+\ln(2/9),\]
which means
\[n\geq \Omega(d_1d_2/\varepsilon^2).\]
\end{proof}

\subsection{Technical lemmas}
\begin{lemma}\label{lemma-12270148}
Let $\ket{\gamma_i}$ be the state defined in \cref{eq-12270131}. Then, $\ketbra{\gamma_i}{\gamma_i}$ is an $n$-comb. This further means $\Gamma_i$ defined in \cref{eq-12272104} is an $n$-comb.
\end{lemma}
\begin{proof}
We use $\ket{\gamma_i^n}$ to denote the state $\ket{\gamma_i}$ defined in \cref{eq-12270131} with parameter $n$. 
Then, we use induction on $n$ to prove $\ketbra{\gamma_i^n}{\gamma_i^n}$ is an $n$-comb. 
First, we note that 
\begin{equation}\label{eq-1120121}
\begin{gathered}
\tr_{\mathcal{H}_\mathrm{B}}(\kettbbra{V_0}{\Delta})=V_0^{\textup{T}}\Delta=0,\quad\quad \tr_{\mathcal{H}_{\mathrm{B}}}(\kettbbra{\Delta}{V_0})=\Delta^{\textup{T}} V_0=0, \\
\tr_{\mathcal{H}_{\mathrm{B}}}(\kettbbra{V_0}{V_0})=I_\mathrm{A},\quad\quad \tr_{\mathcal{H}_{\mathrm{B}}}(\kettbbra{\Delta}{\Delta})=I_\mathrm{A},
\end{gathered}
\end{equation}
where $I_{\mathrm{A}}$ denotes $\sum_{i=1}^{d_1} \ket{i}_\mathrm{A}\bra{i}_{\mathrm{A}}$.
Then, we note that
\[\ket{\gamma_i^i}=\kett{\Delta}^{\otimes i}.\]
Since $\tr_{\mathcal{H}_\mathrm{B}}(\kettbbra{\Delta}{\Delta})=I_{\mathrm{A}}$, we know that $\ketbra{\gamma_i^i}{\gamma_i^i}$ is an $i$-comb, and the hypothesis holds for the case $n=i$.
On the other hand, note that
\[\ket{\gamma_i^n}=\sqrt{\frac{\binom{n-1}{i}}{\binom{n}{i}}}\kett{V_0}\otimes \ket{\gamma_i^{n-1}}+\sqrt{\frac{\binom{n-1}{i-1}}{\binom{n}{i}}}\kett{\Delta}\otimes\ket{\gamma_{i-1}^{n-1}}.\]
Thus, we have
\begin{align}
\tr_{\mathcal{H}_{\mathrm{B},n}}(\ketbra{\gamma_i^n}{\gamma_i^n})&= \frac{\binom{n-1}{i}}{\binom{n}{i}}\tr_{\mathcal{H}_{\mathrm{B}}}(\kettbbra{V_0}{V_0})\otimes \ketbra{\gamma_{i}^{n-1}}{\gamma_i^{n-1}}+\frac{\binom{n-1}{i-1}}{\binom{n}{i}}\tr_{\mathcal{H}_{\mathrm{B}}}(\kettbbra{\Delta}{\Delta})\otimes \ketbra{\gamma_{i-1}^{n-1}}{\gamma_{i-1}^{n-1}}\nonumber\\
&\,\,\,\,+ \frac{\sqrt{\binom{n-1}{i}\binom{n-1}{i-1}}}{\binom{n}{i}}\Big(\tr_{\mathcal{H}_{\mathrm{B}}}(\kettbbra{V_0}{\Delta})\otimes \ketbra{\gamma_{i}^{n-1}}{\gamma_{i-1}^{n-1}}+\tr_{\mathcal{H}_{\mathrm{B}}}(\kettbbra{\Delta}{V_0})\otimes \ketbra{\gamma_{i-1}^{n-1}}{\gamma_{i}^{n-1}}\Big)\nonumber\\
&=\frac{\binom{n-1}{i}}{\binom{n}{i}}I_\mathrm{A}\otimes \ketbra{\gamma_i^{n-1}}{\gamma_i^{n-1}}+\frac{\binom{n-1}{i-1}}{\binom{n}{i}}I_{\mathrm{A}}\otimes \ketbra{\gamma_{i-1}^{n-1}}{\gamma_{i-1}^{n-1}},\label{eq-1120122}
\end{align}
where \cref{eq-1120122} is due to \cref{eq-1120121}.
By induction hypothesis, both $\ketbra{\gamma_i^{n-1}}{\gamma_i^{n-1}}$ and $\ketbra{\gamma_{i-1}^{n-1}}{\gamma_{i-1}^{n-1}}$ are $(n-1)$-combs. 
Note that $\binom{n-1}{i}+\binom{n-1}{i-1}=\binom{n}{i}$ and thus 
\[\frac{\binom{n-1}{i}}{\binom{n}{i}}\ketbra{\gamma_i^{n-1}}{\gamma_i^{n-1}}+\frac{\binom{n-1}{i-1}}{\binom{n}{i}}\ketbra{\gamma_{i-1}^{n-1}}{\gamma_{i-1}^{n-1}}\]
is an $(n-1)$-comb.
Therefore, $\ketbra{\gamma_i^n}{\gamma_{i}^n}$ is an $n$-comb.

Then, $\Gamma_i$ by definition is a convex combination of $n$-comb, thus $\Gamma_i$ is also an $n$-comb.
\end{proof}

\begin{lemma}\label{lemma-12281107}
Suppose $B=2e^4$, $d_1d_2\geq 2$ and $n\leq \frac{1}{B}d_1d_2/\varepsilon^2$. There exists positive numbers $\lambda_0,\ldots,\lambda_n$ such that \cref{eq-12271213} holds and $\sum_{i=0}^n \lambda_i\leq 3d_1^2d_2^2\exp\!\left(\sqrt{8n\varepsilon^2d_1d_2}\right)$.
\end{lemma}
\begin{proof}
From \cref{eq-12271523}, it is easy to see that $\kett{V}^{\otimes n}$ is contained in $\bigoplus_{j=0}^n \supp(\Gamma_j)$. 
Then, by \cref{fact-5122103}, \cref{eq-12271213} is equivalent to 
\begin{equation}\label{eq-12271634}
\sum_{i=0}^n \frac{1}{\lambda_i} \tr\!\left(\Gamma_i^{-1}\kettbbra{V}{V}^{\otimes n}\right)\leq 1.
\end{equation}
Note that the LHS of \cref{eq-12271634} can be upper bounded as
\begin{align}
\sum_{i=0}^n \frac{1}{\lambda_i} \tr\!\left(\Gamma_i^{-1}\kettbbra{V}{V}^{\otimes n}\right)&=\sum_{i=0}^n \frac{1}{\lambda_i}\binom{n}{i}\left(1-\varepsilon^2\right)^{n-i} \varepsilon^{2i} \tr\!\left(\Gamma_i^{-1} U^{\otimes n}\ketbra{\gamma_i}{\gamma_i} U^{\dag\otimes n}\right) \label{eq-12271624} \\
&=\sum_{i=0}^n \frac{1}{\lambda_i}\binom{n}{i}\left(1-\varepsilon^2\right)^{n-i} \varepsilon^{2i} \tr\!\left(\Gamma_i^{-1}\ketbra{\gamma_i}{\gamma_i} \right) \label{eq-12271632}\\
&\leq \sum_{i=0}^n \frac{1}{\lambda_i}\binom{n}{i}\left(1-\varepsilon^2\right)^{n-i} \varepsilon^{2i} \binom{d_1d_2+i-2}{i},\label{eq-12271633}
\end{align}
where \cref{eq-12271624} is by using \cref{eq-12271523} and the fact that $U^{\otimes n}\ket{\gamma_i}$ lies in $\supp(\Gamma_i)$, \cref{eq-12271632} is because $\Gamma_i$ commutes with $U^{\otimes n}$, \cref{eq-12271633} by using \cref{lemma-12271509} where we consider $\ket{\gamma_i}$ as a vector in the linear space:
\[\spanspace\left(\left\{\sum_{\substack{S\subseteq [n]\\ |S|=i }}\ket{\psi}^{\otimes S}\otimes \kett{V_0}^{\otimes [n]\setminus S}\,\, \bigg|\,\, \ket{\psi}\in \kett{V_0}^{\perp} \right\}\right),\]
which has dimension $\binom{d_1d_2+i-2}{i}$ by \cref{lemma-12271712}.
Therefore, it suffices to find positive numbers $\lambda_0,\ldots,\lambda_n$ such that \cref{eq-12271633} is upper bounded by $1$.

Using \cref{lemma-12272204}, we can upper bound \cref{eq-12271633} as
\begin{align}
\eqref{eq-12271633} &\leq \sum_{i=0}^n\frac{1}{\lambda_i} \exp\!\left(-n D\!\left(\frac{i}{n} \,\Big\|\, \varepsilon^2\right) + (d_1d_2+i) H\!\left(\frac{i}{d_1d_2+i}\right)\right) \nonumber\\
&=\sum_{i=0}^n\frac{1}{\lambda_i}\exp\!\left(-i \ln\!\left(\frac{i}{n\varepsilon^2}\right) -(n-i)\ln\!\left(\frac{n-i}{n(1-\varepsilon^2)}\right)+ i \ln \!\left(1+\frac{d_1d_2}{i}\right)+d_1d_2 \ln\!\left(1+\frac{i}{d_1d_2}\right)\right)\nonumber \\
&\leq \sum_{i=0}^n \frac{1}{\lambda_i} \exp\!\left(-i\ln\!\left(\frac{i}{n\varepsilon^2}\right)+i\ln\!\left(1+\frac{d_1d_2}{i}\right)+2i\right),\label{eq-12280023}
\end{align}
where \cref{eq-12280023} is because
\[(n-i)\ln\!\left(\frac{n(1-\varepsilon^2)}{n-i}\right)\leq (n-i)\left(\frac{n(1-\varepsilon^2)}{n-i}-1\right)=i-n\varepsilon^2\leq i,\]
and
\[d_1d_2\ln\!\left(1+\frac{i}{d_1d_2}\right)\leq d_1d_2\frac{i}{d_1d_2}=i.\]
Now, we bound each summand in \cref{eq-12280023} separately:
\begin{itemize}
\item For $i< d_1d_2$, we use
\begin{align}
-i\ln\!\left(\frac{i}{n\varepsilon^2}\right)+i\ln\!\left(1+\frac{d_1d_2}{i}\right)+2i &\leq -i\ln\!\left(\frac{i}{n\varepsilon^2}\right)+i\ln\!\left(\frac{2d_1d_2}{i}\right)+2i \nonumber\\
&=2i\ln\!\left(\frac{\sqrt{2e^2n\varepsilon^2d_1d_2}}{i}\right) \nonumber \\
&\leq \sqrt{8n\varepsilon^2d_1d_2}, \label{eq-12280237}
\end{align}
where \cref{eq-12280237} is due to \cref{lamma-12281043}. 

\item For $i\geq d_1d_2$, then we have $i\geq Bn\varepsilon^2$ since $n\leq \frac{1}{B}d_1d_2/\varepsilon^2$ by assumption, and
\begin{align}
-i\ln\!\left(\frac{i}{n\varepsilon^2}\right)+i\ln\!\left(1+\frac{d_1d_2}{i}\right)+2i&\leq -i \ln(B)+i\ln(2)+2i=-i\ln\!\left(B/2e^2\right)=-2i.\nonumber
\end{align}
\end{itemize}
Therefore, taking $\lambda_i=2d_1d_2\exp\!\left(\sqrt{8n\varepsilon^2d_1d_2}\right)$ for $i<d_1d_2$ and $\lambda_i=\exp(-i)$ for $i\geq d_1d_2$, we can see \cref{eq-12280023} is upper bounded by
\begin{align}
\eqref{eq-12280023} &\leq \sum_{i<d_1d_2} \frac{1}{\lambda_i} \exp\!\left(\sqrt{8n\varepsilon^2d_1d_2}\right)+\sum_{i\geq d_1d_2} \frac{1}{\lambda_i}\exp(-2i) \nonumber \\
&\leq \frac{1}{2}+\sum_{i\geq d_1d_2} \exp(-i)\nonumber \\
&\leq\frac{1}{2}+ \exp(-d_1d_2) \frac{e}{e-1}\nonumber \\
&< 1,\nonumber
\end{align}
where in the last inequality we use that $d_1d_2\geq 2$,
and we also have
\begin{align}
\sum_{i=0}^n\lambda_i &\leq 2d_1^2d_2^2\exp\!\left(\sqrt{8n\varepsilon^2d_1d_2}\right)+\exp(-d_1d_2)\frac{e}{e-1} \nonumber \\
& < 2d_1^2d_2^2\exp\!\left(\sqrt{8n\varepsilon^2d_1d_2}\right)+\frac{1}{2}\nonumber\\
& < 3d_1^2d_2^2\exp\!\left(\sqrt{8n\varepsilon^2d_1d_2}\right),\nonumber
\end{align}
as desired.
\end{proof}

\begin{lemma}\label{lemma-12271509}
Let $G$ be a compact Lie group equipped with a unitary action $\rho(\cdot)$ on a finite-dimensional Hilbert space $\mathcal{H}$.
Let $X\in \mathcal{L}(\mathcal{H})$ be a positive semidefinite operator. Then, we have
\[\tr\left(\left(\E_{g\sim G}\left[\rho(g) X\rho(g)^{-1}\right]\right)^{-1} X\right)\leq \dim(\mathcal{H}),\]
where $(\cdot)^{-1}$ denotes the pseudo-inverse and $\E_{g\sim G}$ is the expectation over the Haar measure of $G$.
\end{lemma}
\begin{proof}
Since $\mathcal{H}$ is a unitary representation of $G$, it is completely reducible. This means we can write: 
\[\mathcal{H}\stackrel{G}{\cong} \bigoplus_{i} \mathcal{V}_i \otimes \mathcal{W}_i,\]
where these $\mathcal{V}_i$ are pairwise non-isomorphic irreducible representations of $G$ and $\mathcal{W}_i$ are corresponding multiplicity spaces. 
We can write $X=\bigoplus_{i,j} X_{i\rightarrow j}$ where $X_{i\rightarrow j}: \mathcal{V}_i\otimes\mathcal{W}_i\rightarrow \mathcal{V}_j\otimes\mathcal{W}_j$ is a linear operator.
Then, by Schur's lemma, we have
\begin{align}
\E_{g\sim G}\left[\rho(g) X\rho(g)^{-1}\right]=\bigoplus_{i} \frac{1}{\dim(\mathcal{V}_i)}I_{\mathcal{V}_i} \otimes \tr_{\mathcal{V}_i}(X_{i\rightarrow i}).\nonumber 
\end{align}
Therefore, 
\begin{align}
\tr\left(\left(\E_{g\sim G}\left[\rho(g) X\rho(g)^{-1}\right]\right)^{-1} X\right)&=\sum_i \dim(\mathcal{V}_i) \tr\Big(\left(I_{\mathcal{V}_i}\otimes \tr_{\mathcal{V}_i}(X_{i\rightarrow i})^{-1}\right) \cdot X_{i\rightarrow i}\Big)\nonumber\\
&=\sum_i \dim(\mathcal{V}_i) \tr\left(\tr_{\mathcal{V}_i}(X_{i\rightarrow i})^{-1}\cdot \tr_{\mathcal{V}_i}(X_{i\rightarrow i})\right)\nonumber\\
&\leq \sum_i \dim(\mathcal{V}_i) \dim(\mathcal{W}_i)\nonumber\\
&=\dim(\mathcal{H}).\nonumber
\end{align}
\end{proof}

\begin{lemma}\label{lemma-12271712}
Consider the linear space $\mathbb{C}^{d+1}$ with the orthonormal basis $\{\ket{0},\ket{1},\ldots,\ket{d}\}$. Let $n\geq m$ be two positive integers and $\mathcal{H}_i\cong \mathbb{C}^{d+1}$ for $i\in[n]$. Consider the following subspace of $\bigotimes_{i=1}^n \mathcal{H}_i$:
\[A=\spanspace\left(\left\{\sum_{\substack{S\subseteq [n]\\ |S|=m }}\ket{\psi}^{\otimes S}\otimes \ket{0}^{\otimes [n]\setminus S}\,\, \bigg|\,\, \ket{\psi}\in \ket{0}^{\perp} \right\}\right),\]
where $\ket{0}^{\perp}$ denotes the subspace orthogonal to $\ket{0}$ (i.e., the subspace spanned by $\{\ket{1},\ldots,\ket{d}\}$).
Then, we have 
\[\dim(A)=\binom{d+m-1}{m}.\]
\end{lemma}
\begin{proof}
Consider the linear operator
\[P=\sum_{\pi\in\mathfrak{S}_{n}} \texttt{p}(\pi),\]
where $\texttt{p}(\cdot)$ denotes the tensor permutation action of $\mathfrak{S}_n$ on $\bigotimes_{i=1}^n\mathcal{H}_i$, i.e., $\texttt{p}(\pi)\ket{\psi_1}\otimes\cdots\otimes\ket{\psi_{n}}=\ket{\psi_{\pi^{-1}(1)}}\otimes\cdots\otimes\ket{\psi_{\pi^{-1}(n)}}$.
One can easily check that $P$ is injective when restricting on the subspace $\spanspace(\{\ket{\psi}^{\otimes m}\otimes \ket{0}^{\otimes n-m} \,\, | \,\, \ket{\psi}\in\ket{0}^\perp\})$ and $A$ is exactly the image of $P$ on this subspace. 
Furthermore, we know that $\spanspace(\{\ket{\psi}^{\otimes m}\,\,|\,\, \ket{\psi}\in\ket{0}^\perp\})\cong \lor^m \mathbb{C}^d$ is the symmetric subspace of $(\mathbb{C}^d)^{\otimes m}$, and has dimension $\binom{d+m-1}{m}$~\cite{harrow2013church}, and thus $A$ has the same dimension.
\end{proof}

\section{Instantiation}\label{sec-1121759}
In \cref{sec-1110117}, we showed that for a sufficiently large set $\mathcal{N}$ of isometries with specific structures, the discrimination problem for $\mathcal{N}$ is hard. 
In this section, we use the construction of the $\varepsilon$-net provided in \cite{oufkir2026improved} as an instantiation of $\mathcal{N}$.
This, combined with our \cref{thm-1110122}, provides the lower bound for quantum channel tomography.

Suppose $r,d_1,d_2$ are positive integers such that $rd_2\geq d_1$ and $\varepsilon\in(0,1)$.
Define the Hilbert spaces $\mathcal{H}_{\mathrm{A}}\cong \mathbb{C}^{d_1}$ and $\mathcal{H}_{\mathrm{B}}\cong \mathbb{C}^{d_2}$, and $\mathcal{H}_{\mathrm{anc}}\cong \mathbb{C}^r$.
Define the isometries $\Delta:\mathcal{H}_{\mathrm{A}}\rightarrow \mathcal{H}_{\mathrm{B}}\otimes\mathcal{H}_{\mathrm{anc}}$ as
\begin{align}
\Delta \coloneqq \sum_{i=1}^{d_1} \ket{i}_{\mathrm{B},\mathrm{anc}}\bra{i}_{\mathrm{A}}. \nonumber
\end{align}
Then, for $U\in\mathbb{U}_{rd_2}$, define the isometries $V_{\varepsilon,U}:\mathcal{H}_{\mathrm{A}}\rightarrow\mathbb{C}^2\otimes \mathcal{H}_{\mathrm{B}}\otimes\mathcal{H}_{\mathrm{anc}}$ as
\begin{align}
V_{\varepsilon,U}&\coloneqq \sqrt{1-\varepsilon^2} \ket{0}\otimes V_0 +\varepsilon \ket{1}\otimes  (U\Delta), \label{eq-1110127}
\end{align}
where $V_0=\sum_{i=1}^r\ket{i}_{\mathrm{anc}}\otimes K_i:\mathcal{H}_{\mathrm{A}}\rightarrow\mathcal{H}_{\mathrm{B}}\otimes\mathcal{H}_{\mathrm{anc}}$ is an isometry such that
\[\left|\tr\!\left(K_i^\dag K_j\right)\right|\leq \frac{2d_1}{r}\cdot \mathbbm{1}_{i=j},\quad\forall i,j\in [r].\]
Then, define the quantum channels $\mathcal{E}_{\varepsilon,U}:\mathcal{L}(\mathcal{H}_{\mathrm{A}})\rightarrow\mathcal{L}(\mathbb{C}^2\otimes \mathcal{H}_{\mathrm{B}})$ as
\begin{equation}\label{eq-1110128}
\mathcal{E}_{\varepsilon,U}(\cdot)\coloneqq \tr_{\mathcal{H}_{\mathrm{anc}}}\!\left(V_{\varepsilon,U}(\cdot) V_{\varepsilon,U}^{\dag}\right).
\end{equation}
The following result is adapted from \cite{oufkir2026improved}.
\begin{lemma}\label{thm-1101611}
Let $d_1,d_2,r$ be positive integers such that $d_1/d_2\leq r \leq d_1d_2$, and $\varepsilon\in (0,10^{-4})$.
There exists a subset $\mathcal{M}\subseteq \mathbb{U}_{rd_2}$ with cardinality $|\mathcal{M}|\geq \exp(rd_1d_2/1201)$ such that for any $U_1,U_2\in\mathcal{M}$ and $U_1\neq U_2$, we have
\[\|\mathcal{E}_{\varepsilon,U_1}-\mathcal{E}_{\varepsilon,U_2}\|_\diamond\geq 0.07\varepsilon,\]
where $\mathcal{E}_{\varepsilon,U}\in\qchannel_{d_1,2d_2}^r$ is defined in \cref{eq-1110128}.
For convenience, we will denote the set of isometries $\{V_{\varepsilon,U} \,|\, U\in\mathcal{M}\}\subseteq \isochannel_{d_1,2rd_2}$ as $\mathcal{N}$. %and the set of channels $\{\mathcal{E}_{\varepsilon,U} \,|\, U\in\mathcal{M}\}\subseteq \qchannel_{d_1,2d_2}^r$ as $\mathcal{N}'$.
\end{lemma}

Then, we can prove the lower bound for quantum channel tomography.
\begin{theorem}\label{thm-1110326}
Let $d_1,d_2,r$ be positive integers such that $2d_1/d_2\leq r\leq d_1d_2/2$.
Suppose $\mathcal{E}\in\qchannel_{d_1,d_2}^r$ is an unknown quantum channel.
Any algorithm that can output an estimate for $\mathcal{E}$ to within diamond norm error $\varepsilon$ with high probability must use at least $n=\Omega(rd_1d_2/\varepsilon^2)$ queries to $\mathcal{E}$.
\end{theorem}
\begin{proof}
If $d_2$ is an even number, we call \cref{thm-1101611} with parameters $(d_1,d_2/2,r)$, %\footnote{Note that this only works when $d_2$ is an even number. However, this is not a big problem, as in \cref{appdix-1140304} we show a modified construction of $\mathcal{N}$ that works for odd $d_2$.}
and we can find a set of isometries $\mathcal{N}\subseteq \isochannel_{d_1,rd_2}$ with cardinality $|\mathcal{N}|\geq \exp(rd_1d_2/C)$ for a universal constant $C$ such that for any $V_1,V_2\in\mathcal{N}$, the channels $\mathcal{E}_1(\cdot)=\tr_{r}\!\left(V_1(\cdot)V_1^\dag\right)$ and $\mathcal{E}_2(\cdot)=\tr_{r}\!\left(V_2(\cdot) V_2^\dag\right)$ satisfy
\[\left\|\mathcal{E}_1-\mathcal{E}_2\right\|_\diamond\geq 0.07\varepsilon.\]
Suppose $\mathcal{A}$ is an algorithm that can output an estimate of an unknown channel in $\qchannel_{d_1,d_2}^r$ to within diamond norm error $0.03\varepsilon$ using $n$ queries to the unknown channel. 
Then, $\mathcal{A}$ can also solve the discrimination task for the isometries in $\mathcal{N}$ using $n$ queries to the unknown isometry by simply discarding the $r$-dimensional ancilla system.

On the other hand, note that the set $\mathcal{N}$ (c.f. \cref{thm-1101611} and \cref{eq-1110127}) is a ``hard'' isometry set (see \cref{def-1130036}).
Thus, \cref{thm-1110122} applies, which means $\mathcal{A}$ must use at least $n\geq \Omega(rd_1d_2/\varepsilon^2)$ queries.

If $d_2$ is an odd number and $r(d_2-1)\geq 2d_1$, then we can simply work with a $(d_2-1)$-dimensional subspace of the output space and find the set $\mathcal{N}$ by calling \cref{thm-1101611} with parameters $(d_1,(d_2-1)/2,r)$, and everything remains the same as those for the even output dimension case.
For the case $d_2$ is an odd number and $r(d_2-1)< 2d_1$, we show a modified construction of $\mathcal{N}$ that works for this case in \cref{appdix-1140304}.
\end{proof}

\section{Auxiliary facts}

We will also use the following well-known facts.
\begin{fact}\label{lemma-12272204}
Let $n\geq k$ be positive integers and $p\in [0,1]$, then
\[\binom{n}{k} \leq \exp\!\left(n H\!\left(k/n\right)\right),\]
and thus
\[\binom{n}{k} p^k (1-p)^{n-k}\leq \exp\!\left(-n D\!\left(k/n \| p\right)\right).\]
\end{fact}
\begin{proof}
Note that $\binom{n}{k}\leq \frac{n^n}{k^k (n-k)^{n-k}}$ since
\[\binom{n}{k} \left(\frac{k}{n}\right)^k\left(1-\frac{k}{n}\right)^{n-k}\leq \left(\frac{k}{n}+1-\frac{k}{n}\right)^n =1.\]
\end{proof}

\begin{fact}\label{fact-5122103}
Suppose $M$ is a positive semidefinite matrix and $\ket{\psi}$ is a vector such that $\ket{\psi}\in\supp(M)$. Then, we have
\[M\sqsupseteq \ketbra{\psi}{\psi}\Longleftrightarrow 1\geq \bra{\psi}M^{-1}\ket{\psi},\]
where $M^{-1}$ is the pseudo-inverse of $M$.
\end{fact}
\begin{proof}
\[ M\sqsupseteq \ketbra{\psi}{\psi}\Longleftrightarrow  I_{\supp(M)} \sqsupseteq M^{-1/2}\ketbra{\psi}{\psi} M^{-1/2},\]
where $M^{-1/2}$ is the pseudo-inverse of $M^{1/2}$. Then
\[ I_{\supp(M)} \sqsupseteq M^{-1/2}\ketbra{\psi}{\psi} M^{-1/2}\Longleftrightarrow 1\geq \tr(M^{-1/2}\ketbra{\psi}{\psi} M^{-1/2}) = \bra{\psi}M^{-1}\ket{\psi}.\]
\end{proof}

\begin{fact}\label{lamma-12281043}
Suppose $x,M> 0$, we have
\[x\ln (M/x) \leq M/e.\]
\end{fact}
\begin{proof}
The derivative of the function $f(x)=x\ln(M/x)$ is $\ln(M/x)-1$ which is monotonically decreasing and equal to zero when $M=ex$. Therefore, the $f(x)\leq f(M/e)=M/e$.
\end{proof}

\bibliographystyle{alpha}
\bibliography{main}

\appendix

\section{Instantiation with odd output dimension}\label{appdix-1140304}

In this section, we present a construction of ``hard'' isometry set which induces an $\varepsilon$-net of quantum channels. The construction follows that given in \cite{oufkir2026improved} (see also \cref{sec-1121759} in our notation), but with a slight modification so that it is adapted to the odd output dimension case.

\subsection{Construction}
Suppose $d_2>1$ is an odd number, $r(d_2-1)< 2d_1$ and $rd_2\geq 2d_1$. 
Let $\mathcal{H}_\mathrm{A}\cong \mathbb{C}^{d_1}$, $\mathcal{H}_{\mathrm{B}}\cong\mathbb{C}^{d_2}$ and $\mathcal{H}_\mathrm{anc}\cong\mathbb{C}^r$ be the input, output and ancilla systems. 
Note that $d_1= \frac{r(d_2-1)}{2} + \eta$ for some integers $1\leq \eta \leq \lfloor\frac{r}{2}\rfloor$. 
Then, we consider the following decomposition
\[\mathcal{H}_\mathrm{A} = \mathcal{H}_{\mathrm{a}}\oplus\mathcal{H}'_{\mathrm{a}},\]
\[\mathcal{H}_\mathrm{B} = \mathcal{H}_{\mathrm{b}0} \oplus \mathcal{H}_{\mathrm{b}1} \oplus \mathcal{H}_{\mathrm{b}}',\]
where 
\[\mathcal{H}_{\mathrm{a}}=\spanspace\left\{\ket{i}_\mathrm{A}\,\bigg|\, 1\leq i \leq \frac{r(d_2-1)}{2}\right\},\]
\[\mathcal{H}_{\mathrm{a}}'=\spanspace\left\{\ket{i}_\mathrm{A}\,\bigg|\, \frac{r(d_2-1)}{2}+1\leq i\leq d_1\right\},\]
\[\mathcal{H}_\mathrm{b0}=\spanspace\left\{\ket{i}_\mathrm{B}\,\bigg|\, 1\leq i \leq \frac{d_2-1}{2}\right\},\]
\[\mathcal{H}_\mathrm{b1}=\spanspace\left\{\ket{i}_\mathrm{B}\,\bigg|\, \frac{d_2+3}{2}\leq i \leq d_2\right\},\]
\[\mathcal{H}_\mathrm{b}'=\spanspace\left\{\ket{i}_\mathrm{B}\,\bigg|\, i= \frac{d_2+1}{2}\right\}.\]
Then we further consider the following decomposition
\[\mathcal{H}_\mathrm{b}'\otimes \mathcal{H}_\mathrm{anc}= \mathcal{H}_{\mathrm{b}0}'\oplus\mathcal{H}_{\mathrm{b}1}',\]
where
\[\mathcal{H}_{\mathrm{b}0}'=\spanspace\left\{\ket{i}_\mathrm{B}\otimes\ket{j}_\mathrm{anc}\,\bigg|\, i= \frac{d_2+1}{2}, \,\, 1 \leq j\leq \left\lfloor\frac{r}{2}\right\rfloor\right\},\]
\[\mathcal{H}_{\mathrm{b}1}'=\spanspace\left\{\ket{i}_\mathrm{B}\otimes\ket{j}_\mathrm{anc}\,\bigg|\, i= \frac{d_2+1}{2}, \,\, \left\lfloor\frac{r}{2}\right\rfloor+1 \leq j\leq r\right\}.\]
Then, note that
\[\dim(\mathcal{H}_\mathrm{a})=\frac{r(d_2-1)}{2},\quad \dim(\mathcal{H}_{\mathrm{a}}')=\eta\leq \left\lfloor\frac{r}{2}\right\rfloor,\]
\[\dim(\mathcal{H}_{\mathrm{b}0})=\dim(\mathcal{H}_{\mathrm{b}1})=\frac{(d_2-1)}{2},\quad \dim(\mathcal{H}_{\mathrm{b}}')=1\]
\[\dim(\mathcal{H}'_{\mathrm{b}0})=\left\lfloor\frac{r}{2}\right\rfloor,\quad \dim(\mathcal{H}_{\mathrm{b}1}')=r-\left\lfloor\frac{r}{2}\right\rfloor\geq \left\lfloor\frac{r}{2}\right\rfloor.\]

Define the isometry $V_0:\mathcal{H}_\mathrm{a}\rightarrow  \mathcal{H}_\mathrm{b0}\otimes \mathcal{H}_\mathrm{anc}$ as
\begin{equation}\label{eq-1131347}
V_0=\sum_{i=1}^r \ket{i}_\mathrm{anc}\otimes K_i,
\end{equation}
such that $K_i:\mathcal{H}_\mathrm{a}\rightarrow\mathcal{H}_{\mathrm{b}0}$ satisfy
\[\left|\tr\!\left(K_i^\dag K_j\right)\right|\leq \frac{2\dim(\mathcal{H}_\mathrm{a})}{\dim(\mathcal{H}_{\mathrm{anc}})}\cdot \mathbbm{1}_{i=j},\]
where the existence of such isometry is proven in \cite[Appendix B]{oufkir2026improved}.
Define $V_0':\mathcal{H}_\mathrm{a}'\rightarrow\mathcal{H}_{\mathrm{b}0}'$ be an arbitrary isometry.
%\[V_0'\ket{i}_\mathrm{A} = \left|\frac{d_2+1}{2}\right\rangle_\mathrm{B}\otimes\left|i-\frac{r(d_2-1)}{2}\right\rangle_\mathrm{anc},\]
%for $\frac{r(d_2-1)}{2}+1\leq i\leq d_1$.
Define $\Delta:\mathcal{H}_\mathrm{a}\rightarrow \mathcal{H}_{\mathrm{b}1}\otimes \mathcal{H}_\mathrm{anc}$ be an arbitrary isometry. 
Define $\Delta':\mathcal{H}_\mathrm{a}'\rightarrow \mathcal{H}'_{\mathrm{b}1}$
be an arbitrary isometry.
Then, for $\varepsilon\in(0,1)$ and $U\in\mathbb{U}_{r(d_2-1)/2}$, we define the isometry $V_{\varepsilon,U}:\mathcal{H}_\mathrm{A}\rightarrow\mathcal{H}_\mathrm{B}\otimes\mathcal{H}_{\mathrm{anc}}$ as
\begin{equation}\label{eq-1140331}
V_{\varepsilon,U}\coloneqq \sqrt{1-\varepsilon^2}(V_0+ V_0')+\varepsilon (U\Delta+ \Delta'),
\end{equation}
where $U$ acts on $\mathcal{H}_{\mathrm{b}1}\otimes\mathcal{H}_\mathrm{anc}$ and $V_0+V_0'=V_0\oplus V_0'$, $U\Delta +\Delta'=U\Delta\oplus \Delta'$ are both direct sums of linear operators.
Moreover, the image of $U\Delta+\Delta'$ (i.e., $(\mathcal{H}_{\mathrm{b}1}\otimes\mathcal{H}_{\mathrm{anc}})\oplus \mathcal{H}_{\mathrm{b}1}'$) is orthogonal to the image of $V_0+V_0'$ (i.e., $(\mathcal{H}_{\mathrm{b}0}\otimes\mathcal{H}_{\mathrm{anc}})\oplus \mathcal{H}_{\mathrm{b}0}'$). Therefore, any subset of $\{V_{\varepsilon,U}\,|\, U\in\mathbb{U}_{r(d_2-1)/2}\}$ is a ``hard'' isometry set.

\subsection{Existence}
We have shown that any subset of $\{V_{\varepsilon,U}\,|\, U\in\mathbb{U}_{r(d_2-1)/2}\}$ (see \cref{eq-1140331}) is a ``hard'' isometry set. Then, we prove that there exists a large subset with good separation property.
\begin{theorem}\label{thm-1140332}
There exists a finite subset $\mathcal{N}$ of $\{V_{\varepsilon,U}\,|\, U\in\mathbb{U}_{r(d_2-1)/2}\}$ for $V_{\varepsilon,U}$ defined in \cref{eq-1140331} with cardinality $|\mathcal{N}|\geq \exp(rd_1d_2/100001)$, such that for any $V_1\neq V_2\in\mathcal{N}$ we have 
\[\|\tr_{\mathcal{H}_\mathrm{anc}}(V_1(\cdot)V_1^\dag)-\tr_{\mathcal{H}_\mathrm{anc}}(V_2(\cdot)V_2^\dag)\|_\diamond\geq 0.07\varepsilon.\]
\end{theorem}
\begin{proof}
The proof follows essentially the argument in \cite{oufkir2026improved}, with slight modifications.
First, we need the following lemma.
\begin{lemma}\label{lemma-1131330}
There exists a finite subset $\mathcal{M}\subseteq \mathbb{U}_{r(d_2-1)/2}$ with cardinality $|\mathcal{M}|\geq \exp(rd_1d_2/100001)$ such that for any $U_1\neq U_2\in\mathcal{M}$, 
\begin{equation}\label{eq-1130328}
\frac{1}{d_1}\left\|\tr_{\mathcal{H}_\mathrm{anc}}\!\left((\kett{V_0}+\kett{V_0'})\Big(\bbra{U_1\Delta}-\bbra{U_2\Delta}\Big)\right)\right\|_1\geq 0.05.
\end{equation}
\end{lemma}

Define the quantum channel $\mathcal{E}_{\varepsilon,U}: \mathcal{L}(\mathcal{H}_{\mathrm{A}})\rightarrow\mathcal{L}(\mathcal{H}_{\mathrm{B}})$ as 
\[\mathcal{E}_{\varepsilon,U}(\cdot)\coloneqq \tr_{\mathcal{H}_{\mathrm{anc}}}\!\left(V_{\varepsilon,U}(\cdot) V_{\varepsilon,U}^\dag\right).\]
Let $C_\mathcal{E}$ denote the (unnormalized) Choi state of quantum channel $\mathcal{E}$ and $\mathcal{M}$ be the set given in \cref{lemma-1131330}.
Then, for any $U_1\neq U_2\in\mathcal{M}$, we have
\begin{align}
\left\|C_{\mathcal{E}_{\varepsilon,U_1}}-C_{\mathcal{E}_{\varepsilon,U_2}}\right\|_1 &=\Big\|\tr_{\mathcal{H}_{\mathrm{anc}}}\!\left(\kettbbra{V_{\varepsilon,U_1}}{V_{\varepsilon,U_1}}\right)-\tr_{\mathcal{H}_{\mathrm{anc}}}\!\left(\kettbbra{V_{\varepsilon,U_2}}{V_{\varepsilon,U_2}}\right)\Big\|_1 \nonumber\\
&=\bigg\|\varepsilon^2\Big(\tr_{\mathcal{H}_{\mathrm{anc}}}\!\left(\kettbbra{\Delta_{U_1}}{\Delta_{U_1}}\right)-\tr_{\mathcal{H}_{\mathrm{anc}}}\!\left(\kettbbra{\Delta_{U_2}}{\Delta_{U_2}}\right)\Big) \label{eq-1130257}\\
&\quad\quad\quad +\varepsilon\sqrt{1-\varepsilon^2} \tr_{\mathcal{H}_{\mathrm{anc}}}\!\Big(\kettbbra{V_0+V_0'}{(U_1-U_2)\Delta}\Big) \nonumber \\
&\quad\quad\quad + \varepsilon\sqrt{1-\varepsilon^2} \tr_{\mathcal{H}_{\mathrm{anc}}}\!\Big(\kettbbra{(U_1-U_2)\Delta}{V_0+V_0'}\Big) \bigg\|_1 \nonumber \\
&\geq \varepsilon\sqrt{1-\varepsilon^2}\bigg\|\tr_{\mathcal{H}_{\mathrm{anc}}}\!\Big(\kettbbra{V_0+V_0'}{(U_1-U_2)\Delta}\Big) \nonumber \\
&\qquad \qquad\qquad\qquad+ \tr_{\mathcal{H}_{\mathrm{anc}}}\!\Big(\kettbbra{(U_1-U_2)\Delta}{V_0+V_0'}\Big)\bigg\|_1 - 2 \varepsilon^2 d_1 \label{eq-1130259} \\
& = 2\varepsilon\sqrt{1-\varepsilon^2} \bigg\|\tr_{\mathcal{H}_{\mathrm{anc}}}\!\Big(\kettbbra{V_0+V_0'}{(U_1-U_2)\Delta}\Big)\bigg\|_1-2\varepsilon^2 d_1 \label{eq-1130304}\\
&\geq 0.1\varepsilon\sqrt{1-\varepsilon^2}d_1 -2\varepsilon^2d_1 \label{eq-1130327} \\
&\geq 0.07\varepsilon d_1. \label{eq-1130422}
\end{align}
In \cref{eq-1130257} we define $\Delta_U=U\Delta+\Delta'$. In \cref{eq-1130259} we used 
\[\|\tr_{\mathcal{H}_{\mathrm{anc}}}(\kettbbra{\Delta_U}{\Delta_U})\|_1 = \tr(\kettbbra{\Delta_U}{\Delta_U})=\dim(\mathcal{H}_\mathrm{a})+\dim(\mathcal{H}'_\mathrm{a})=d_1. \]
In \cref{eq-1130304} we used the fact that $\tr_{\mathcal{H}_{\mathrm{anc}}}\!\Big(\kettbbra{V_0+V_0'}{(U_1-U_2)\Delta}\Big)$ is a linear operator supported on  $\mathcal{\mathcal{H}_\mathrm{A}}\otimes\mathcal{H}_{\mathrm{b1}}$ and its image is in $\mathcal{H}_\mathrm{A}\otimes (\mathcal{H}_{\mathrm{b}0}\oplus \mathcal{H}_\mathrm{b}')$, which is orthogonal to $\mathcal{\mathcal{H}_\mathrm{A}}\otimes\mathcal{H}_{\mathrm{b1}}$; and for any linear operator $X$ such that $X^2=0$, we have 
\[\|X+X^\dag\|_1=\tr\!\left(\sqrt{(X+X^\dag)^2}\right)=\tr\!\left(\sqrt{XX^\dag+X^\dag X}\right)=\tr\!\left(\sqrt{XX^\dag}\right)+\tr\!\left(\sqrt{X^\dag X}\right),\]
where the last equality is because $\supp(XX^\dag)\perp\supp(X^\dag X)$.
In \cref{eq-1130327} we used \cref{eq-1130328}.
In \cref{eq-1130422} we used that $\varepsilon\leq 0.01$.
Therefore, we can lower bound the diamond norm 
\[\left\|\mathcal{E}_{\varepsilon,U_1}-\mathcal{E}_{\varepsilon,U_2}\right\|_\diamond\geq \frac{1}{d}\left\|C_{\mathcal{E}_{\varepsilon,U_1}}-C_{\mathcal{E}_{\varepsilon,U_2}}\right\|_1\geq 0.07\varepsilon.\]
Thus, the set $\mathcal{N}=\{V_{\varepsilon,U} \,|\, U\in\mathcal{M}\}$ is the desired set.
\end{proof}

\subsection{\texorpdfstring{Proof of Lemma \ref{lemma-1131330}}{Proof of Lemma A.2}}
\begin{proof}
Let $\hat{V}_0\coloneqq V_0+V_0'$. Note that $\hat{V}_0: \mathcal{H}_\mathrm{A}\rightarrow (\mathcal{H}_{\mathrm{b}0}\oplus\mathcal{H}_{\mathrm{b}}')\otimes\mathcal{H}_{\mathrm{anc}}$ is an isometry.
The isometry $\hat{V}_0$ can also be written as
\[\hat{V}_0=\sum_{i=1}^r \ket{i}_{\mathrm{anc}}\otimes \left(K_i\oplus \left|\frac{d_2+1}{2}\right\rangle\!\bra{z_i}\right),\]
where $K_i$ are defined in \cref{eq-1131347}, $\{\ket{z_i}\}_{i=1}^{\eta}$ are an orthonormal basis of $\mathcal{H}_\mathrm{a}'$ and $\ket{z_i}=0$ for $\eta< i\leq r$.
Let us define 
\[K_i'=K_i\oplus \left|\frac{d_2+1}{2}\right\rangle\!\bra{z_i},\]
then $K_i'$ satisfy
\begin{equation}\label{eq-1132146}
\left|\tr\!\left(K_i'^\dag K'_j\right)\right|= \left|\tr\!\left(K_i^\dag K_j\right)+\mathbbm{1}_{i=j\in[\eta]}\right|\leq \left(\frac{2\dim(\mathcal{H}_\mathrm{a})}{\dim(\mathcal{H}_{\mathrm{anc}})}+1\right)\cdot \mathbbm{1}_{i=j}\leq \frac{3\dim(\mathcal{H}_\mathrm{A})}{\dim(\mathcal{H}_{\mathrm{anc}})}\cdot \mathbbm{1}_{i=j}=\frac{3d_1}{r}\cdot \mathbbm{1}_{i=j},
\end{equation}
where in the last inequality we used that $\dim(\mathcal{H}_\mathrm{a})=r(d_2-1)/2\geq r=\dim(\mathcal{H}_{\mathrm{anc}})$ and $\dim(\mathcal{H}_\mathrm{a})\leq \dim(\mathcal{H}_{\mathrm{A}})$.
On the other hand, $U\Delta$ is an isometry from $\mathcal{H}_\mathrm{a}\subseteq\mathcal{H}_\mathrm{A}$ to $\mathcal{H}_{\mathrm{b}1}\otimes\mathcal{H}_{\mathrm{anc}}$.

Then, we need the following lemma:
\begin{lemma}\label{lemma-1140127}
For $U_x,U_y\in\mathbb{U}_{r(d_2-1)/2}$, let us define 
\[F(U_x,U_y)=\frac{1}{d_1}\tr_{\mathcal{H}_\mathrm{anc}}\!\left(\kett{\hat{V}_0}\Big(\bbra{U_x\Delta}-\bbra{U_y\Delta}\Big)\right),\]
then the function $f(U_x,U_y)=\|F(U_x,U_y)\|_1=\tr(|F(U_x,U_y)|)$ is $\sqrt{\frac{2}{d_1}}$-Lipschitz with respect to the $\ell_2$-sum of the $2$-norms (Frobenius norm). 
Furthermore, for independent random $U_x,U_y\sim \mathbb{U}_{r(d_2-1)/2}$, we have $\E\!\left[\tr(|F(U_x,U_y)|^2)\right]=\frac{d_2-1}{d_1}$, and $\E\!\left[\tr(|F(U_x,U_y)|^4)\right]\leq \frac{288}{r^3}$.
\end{lemma}

By the H\"older's inequality we have
\[\E\!\left[\tr\!\left(|F(U_x,U_y)|^2\right)\right]\leq \E\!\left[\tr\!\left(|F(U_x,U_y)|^4\right)\right]^{1/3}\E\!\left[\tr\!\left(|F(U_x,U_y)|\right)\right]^{2/3},\]
which, combined with \cref{lemma-1140127}, implies
\[\E\!\left[\tr\!\left(|F(U_x,U_y)|\right)\right]^2\geq \frac{(d_2-1)^3r^3}{288 \cdot d_1^3}\geq \frac{2^3\cdot 8}{3^3\cdot 288}=\frac{2}{243},\]
where we used that $d_2-1\geq 2d_2/3$ and $rd_2\geq 2d_1$. Thus $\E\!\left[\tr\!\left(|F(U_x,U_y)|\right)\right]> 9/100$.
Then, we can use a generalized Levy's lemma on compact groups~\cite[Corollary 17]{10.1214/ECP.v18-2551} to prove the concentration result:
\[\Pr\!\left[\tr\!\left(|F(U_x,U_y)|\right)\leq \frac{1}{20}\right]\leq \exp\!\left(-\frac{r(d_2-1)}{2}\cdot \frac{d_1}{25^2\cdot 12\cdot 2}\right)=\exp\!\left(-\frac{rd_1(d_2-1)}{30000}\right)\leq \exp\!\left(-\frac{rd_1d_2}{50000}\right).\]
Then, we independently sample $\exp(rd_1d_2/100001)$ Haar random unitaries in $\mathbb{U}_{r(d_2-1)/2}$ and the union bound shows that there exists a non-zero probability that for any pair $U_x, U_y$, we have $\tr(|F(U_x,U_y)|)\geq 1/20$.
Thus, there exists a set with cardinality $\geq \exp(rd_1d_2/100001)$ such that \cref{eq-1130328} holds.
\end{proof}

\subsection{\texorpdfstring{Proof of \cref{lemma-1140127}}{Proof of Lemma A.3}}
\begin{proof}
For the Lipschitz continuity, the proof is the same as that given in \cite{oufkir2026improved}.

Define $K_{x,i}=\bra{i}_\mathrm{anc} U_x\Delta$, $K_{y,i}=\bra{i}_\mathrm{anc} U_y\Delta$. This means 
\[F(U_x,U_y)=\frac{1}{d_1}\sum_{i=1}^r\kett{K_i'}\Big(\bbra{K_{x,i}}-\bbra{K_{y,i}}\Big).\]

Then, we note that
\begin{align}
\E\!\left[\tr(|F(U_x,U_y)|^2)\right]&=\frac{1}{d_1^2}\E\!\left[\tr\!\left(\sum_{i,j=1}^r \kett{K_i'}\Big(\bbra{K_{x,i}}-\bbra{K_{y,i}}\Big)\Big(\kett{K_{x,j}}-\kett{K_{y,j}}\Big)\bbra{K_j'}\right)\right] \nonumber\\
&=\frac{1}{d_1^2}\E\!\left[\sum_{i}^r \bbrakett{K_i'}{K_i'}\Big(\bbra{K_{x,i}}-\bbra{K_{y,i}}\Big)\Big(\kett{K_{x,i}}-\kett{K_{y,i}}\Big)\right]\nonumber \\
&=\frac{2}{d_1^2}\sum_{i=1}^r\bbrakett{K_i'}{K_i'} \frac{\dim(\mathcal{H}_\mathrm{a})}{r}\label{eq-1131627} \\
&=\frac{2}{d_1^2}\frac{\dim(\mathcal{H}_\mathrm{A})\dim(\mathcal{H}_\mathrm{a})}{r} \label{eq-1131700}\\
&= \frac{d_2-1}{d_1} \nonumber 
\end{align}
where \cref{eq-1131627} is because for $z_1,z_2\in\{x,y\}$, we have
\begin{align}
\E\!\left[\bbrakett{K_{z_1,i}}{K_{z_2,i}}\right]&=\E\!\left[\tr\!\left(K_{z_1,i}^\dag K_{z_2,i}\right)\right] =\tr\!\left(\Delta^\dag \E\!\left[U_{z_1}^\dag \ket{i}_\mathrm{anc} \bra{i}_\mathrm{anc}U_{z_2} \right]\Delta\right) \nonumber\\
&=\mathbbm{1}_{z_1=z_2}\frac{1}{r}  \tr(\Delta^\dag \Delta) \label{eq-1131652}\\
&=\mathbbm{1}_{z_1=z_2}\frac{\dim(\mathcal{H}_\mathrm{a})}{r}, \nonumber
\end{align}
where \cref{eq-1131652} is due to Schur's lemma, and \cref{eq-1131700} is because $\sum_{i} \bbrakett{K_i'}{K_i'}=\tr(\sum_i K_i'^\dag K_i)=\tr(I_{\mathrm{A}})=\dim(\mathcal{H}_\mathrm{A})$.

Then, we 
\begin{align}
\E\!\left[\tr(|F(U_x,U_y)|^4)\right]&=\frac{1}{d_1^4}\E\!\Bigg[\sum_{i,j,k,l=1}^r \tr\!\bigg(\kett{K_i'}\Big(\bbra{K_{x,i}}-\bbra{K_{y,i}}\Big)\Big(\kett{K_{x,j}}-\kett{K_{y,j}}\Big)\bbra{K'_j} \\
&\qquad\qquad\qquad\qquad\kett{K'_{k}}\Big(\bbra{K_{x,k}}-\bbra{K_{y,k}}\Big)\Big(\kett{K_{x,l}}-\kett{K_{y,l}}\Big)\bbra{K'_{l}}\bigg)\Bigg]  \nonumber \\ 
%&= \frac{1}{d_1^4}\E\!\Bigg[\sum_{i,j=1}^r \bbrakett{K_i'}{K_i'}\bbrakett{K_j'}{K_j'}\Big(\bbra{K_{x,i}}-\bbra{K_{y,i}}\Big)\Big(\kett{K_{x,j}}-\kett{K_{y,j}}\Big) \nonumber\\
%&\qquad\qquad\qquad\qquad\qquad\Big(\bbra{K_{x,j}}-\bbra{K_{y,j}}\Big)\Big(\kett{K_{x,i}}-\kett{K_{y,i}}\Big)\Bigg] \nonumber\\
&\leq \frac{9}{r^2d_1^2}\sum_{i,j=1}^r\E\!\bigg[\left|\Big(\bbra{K_{x,i}}-\bbra{K_{y,i}}\Big)\Big(\kett{K_{x,j}}-\kett{K_{y,j}}\Big)\right|^2\bigg]\label{eq-1132143} \\
&\leq \frac{36}{r^2d_1^2}\sum_{i,j=1}^r\E\!\left[|\bbrakett{K_{x,i}}{K_{x,j}}|^2+|\bbrakett{K_{y,i}}{K_{y,j}}|^2+|\bbrakett{K_{x,i}}{K_{y,j}}|^2+|\bbrakett{K_{y,i}}{K_{x,j}}|^2\right] \nonumber\\
&\leq \frac{72}{r^2d_1^2}\sum_{i,j=1}^r\E\!\left[|\bbrakett{K_{x,i}}{K_{x,j}}|^2+|\bbrakett{K_{y,i}}{K_{y,j}}|^2\right]\label{eq-1132347}\\
&=\frac{144}{r^2d_1^2}\sum_{i,j=1}^r \E\!\left[|\bbrakett{K_{x,i}}{K_{x,j}}|^2\right],\nonumber\\
&\leq \frac{144}{r^2d_1^2} \sum_{i,j=1}^r \frac{d_\mathrm{a}}{(rd_{\mathrm{b}})^2-1}\left(d_\mathrm{b}+\mathbbm{1}_{i=j}d_\mathrm{a}d_\mathrm{b}^2-\frac{1}{d_\mathrm{b}r}\left(\mathbbm{1}_{i=j}d_\mathrm{b}^2+d_\mathrm{a}d_\mathrm{b}\right)\right)\label{eq-1140038}\\
&\leq \frac{144}{r^2}\cdot \frac{1}{d_\mathrm{a}}\cdot \frac{1}{(rd_\mathrm{b})^2-1}\left(r^2d_\mathrm{b}+rd_\mathrm{a}d_\mathrm{b}^2-d_\mathrm{b}-rd_\mathrm{a}\right) \label{eq-1140111}\\
&= \frac{144}{r^2} \left(\frac{1}{d_\mathrm{a}d_\mathrm{b}}+\frac{1}{r}+\frac{1-d_\mathrm{b}^2+d_\mathrm{a}d_\mathrm{b}/r-d_\mathrm{a}d_\mathrm{b}r}{d_\mathrm{a}d_\mathrm{b}(r^2d_\mathrm{b}^2-1)}\right)\nonumber \\
&\leq \frac{288}{r^3},\label{eq-1140112}
\end{align}
where \cref{eq-1132143} is because \cref{eq-1132146}, \cref{eq-1132347} is because 
\begin{align}
\sum_{i,j=1}^r |\bbrakett{K_{x,i}}{K_{y,j}}|^2&= \|K_x^\dag K_y\|_F^2 =\tr\!\left(K_x^\dag K_y K_y^\dag K_x\right)\leq \|K_xK_x^\dag\|_F \cdot \|K_yK_y^\dag\|_F \nonumber \\
&\leq \frac{1}{2}\left(\|K_x^\dag K_x\|_F^2 + \|K_y^\dag K_y\|_F^2\right)=\frac{1}{2}\left(\sum_{i,j=1}^r |\bbrakett{K_{x,i}}{K_{x,j}}|^2 +\sum_{i,j=1}^r |\bbrakett{K_{y,i}}{K_{y,j}}|^2 \right)\nonumber,
\end{align}
where $K_x$ denotes the matrix with columns $\kett{K_{x,i}}$,
\cref{eq-1140038} is due to exactly the same argument as that in Eq. (79) in \cite{oufkir2026improved} and we set $d_\mathrm{a}=\dim(\mathcal{H}_\mathrm{a})$, $d_{\mathrm{b}}=\dim(\mathcal{H}_{\mathrm{b}1})$, 
\cref{eq-1140111} uses $d_1\geq d_\mathrm{a}$ and \cref{eq-1140112} uses $d_\mathrm{a}d_\mathrm{b}=r(d_2-1)^2/4\geq r$.
\end{proof}
\end{document}